\documentclass[12pt]{article}

\usepackage{a4wide}
\usepackage[T2A]{fontenc}
\usepackage[utf8]{inputenc}
\usepackage{amssymb, amsthm}
\usepackage{amsmath}
\usepackage{comment}
\usepackage{color}
\usepackage[pdftex]{graphicx}
\usepackage{hyperref}

\DeclareMathOperator{\Dual}{asSeries}

\DeclareMathOperator{\withdiv}{WithDiv}
\DeclareMathOperator{\withinv}{WithInvDiv}

\begin{document}

\def\newstr{\par\noindent}
\def\ms{\medskip\newstr}

\renewcommand{\proofname}{Proof}

\newcommand{\ifdraft}[1]{#1}
\definecolor{aocolour}{rgb}{0.7,0.8,1}
\definecolor{vmcolour}{rgb}{1,0.8,0.7}
\newcommand{\ao}[1]{\ifdraft{\noindent\colorbox{aocolour}{A.O.: #1}}}
\newcommand{\vm}[1]{\ifdraft{\noindent\colorbox{vmcolour}{V.M.: #1}}}

\newcommand{\Z}{\mathbb{Z}}
\newcommand{\N}{\mathbb{N}}
\newcommand{\R}{\mathbb{R}}
\newcommand{\Q}{\mathbb{Q}}
\newcommand{\K}{\mathbb{K}}
\newcommand{\Cm}{\mathbb{C}}
\newcommand{\Pm}{\mathbb{P}}
\newcommand{\Zero}{\mathbb{O}}
\newcommand{\F}{\mathbb{F}_2}
\newcommand{\ilim}{\int\limits}
\newcommand{\impl}{\Rightarrow}
\newcommand{\set}[2]{\{ \, #1 \mid #2 \, \}}

\newcommand{\A}{\mathcal{A}}
\newcommand{\B}{\mathcal{B}}

\theoremstyle{plain}
\newtheorem{thm}{Theorem}[section]
\newtheorem{lm}{Lemma}[section]
\newtheorem{conjecture}{Conjecture}
\newtheorem*{st}{Statement}
\newtheorem*{prop}{Properties}

\newtheorem{oldtheorem}{Theorem}
\renewcommand{\theoldtheorem}{\Alph{oldtheorem}}

\theoremstyle{definition}
\newtheorem*{defn}{Definition}
\newtheorem{ex}{Example}[section]
\newtheorem*{cor}{Corollary}

\theoremstyle{remark}
\newtheorem*{rem}{Remark}
\newtheorem*{note}{Note}

\title{Finer characterization of bounded languages described by GF(2)-grammars}
\author{Vladislav Makarov, Marat Movsin \\ Saint Petersburg State University}

\maketitle

\begin{abstract}
GF(2)-grammars are a somewhat recently introduced grammar family that have some unusual 
algebraic properties and are closely connected to unambiguous grammars. In ``Bounded languages described by GF(2)-grammars'', Makarov proved a necessary condition for subsets of $a_1^* a_2^* \cdots a_k^*$ to be described by some GF(2)-grammar. By extending these methods further, we 
prove an even stronger upper bound for these languages. Moreover, we establish a lower bound
that closely matches the proven upper bound. Also, we prove the exact characterization for the special case of \emph{linear} GF(2)-grammars. Finally, by using the previous result, we show that the class of languages described by linear GF(2)-grammars is not closed under GF(2)-concatenation.


\end{abstract}

\sloppy

\section{Introduction}
	
The paper assumes that you have already read and understood the previous 
paper about the bounded languages and GF(2)-grammars by Makarov~\cite{Makarov_DLT}
(an extended version of that paper is easily available on arXiv). 
For a general introduction to the topic and the methods, refer to said paper.

Reading the earlier
papers (the original paper aboud GF(2)-operations by Bakinova et al.~\cite{gf2} and
the paper about basic properties of GF(2)-grammars by Makarov and Okhotin~\cite{grammars-gf2})
is not required, but may help with understanding the context of some of our results.

\section{A general upper bound for subsets of \texorpdfstring{$a_1^* a_2^* \ldots a_k^*$}
{a1* a2* ... ak*}}

Similarly to $R_{a, b}$, $R_{a, b, c}$ and $R_{a_1, a_2, \ldots, a_k}$ 
from the previous paper we want to define some kind of a ring structure. 
Specifically, we want to consider something like $R_{a_1, a_2, \ldots, a_k}$,
but with only some pairs $(a_i, a_j)$ being allowed to define a polynomial
in the numerator. Let us state what we need precisely.

Suppose that we are working over the alphabet $\Sigma = \{a_1, a_2, \ldots, a_k\}$
and are using the natural correspondence between subsets of $\Sigma^*$ and the series from
$\F[[a_1, a_2, \ldots, a_k]]$, defined by $\Dual$. As a bit of abuse of notation and language, we will sometimes
talk about the series as if they were languages and vice versa. Then, consider the following
definition.

\begin{defn} Suppose that $T$ is some subring of $\F[[a_1, a_2, \ldots, a_k]]$ that contains
all polynomials. In other words, $\F[a_1, a_2, \ldots a_k] \subset T \subset \F[[a_1, a_2, \ldots, a_k]]$. Consider the set $X_k := \{(1, 2), (1, 3), \ldots, (1, k), (2, 3), \ldots, (2, k), \ldots, (k - 1, k)\}$ of all all pairs $(i, j)$, such that $1 \leqslant i < j \leqslant n$. Then, for every subset $Y$ of $X_k$, 
we can define a subset $\withdiv(T, Y)$  of $\F[[a_1, a_2, \ldots, a_k]]$ in the following way. 
An element of $\F[[a_1, a_2, \ldots, a_k]]$ lies in $\withdiv(T, Y)$ if and only if it can 
be represented as $\dfrac{t}{\prod_{(i, j) \in Y} p_{(i, j)}}$
for some $t \in T$ and $p_{(i, j)} \in \F[a_i, a_j]$.
\end{defn}
\begin{ex}
For example, $R_{a_1, a_2, \ldots, a_k}$ is $\withdiv(S_k, X_k)$ for some ring $S_k$. 
\end{ex}
\begin{defn}
More precisely, let $S_k$ be the set 
of all power series that can be represented as $\sum_{i=1}^n A_{i, 1} A_{i, 2} \ldots A_{i, k}$
for some $n \geqslant 0$ and $A_{i, j} \in \A_j$. Or, in other words, $S_k$ is the 
subring of $\F[[a_1, a_2, \ldots, a_k]]$ generated by all algebraic power series in one variable.
\end{defn}
\begin{ex} As a more complicated example, consider $k = 3$ and the set $\withdiv(\F[a_1, a_2, a_3], \{(1, 2), (2, 3)\})$. Said set consists of all power series that can be represented as $p/(qr)$ for some polynomials $p \in \F[a, b, c]$, $q \in \F[a, b]$ and $r \in \F[b, c]$.
\end{ex}
\begin{rem} It is important to note that $\withdiv(T, Y)$ includes only power series. 
For example, $\withdiv(\F[a_1, a_2], \{(1,2)\})$ is not the set of rational functions $\F(a_1, a_2)$,
but the set $\F(a_1, a_2) \cap \F[[a_1, a_2]]$.
\end{rem}

Let us prove the following important lemma:

\begin{lm} The set $\withdiv(T, Y)$ is actually a ring.
\end{lm}
\begin{proof} It it easy to check that this set is closed under all necessary operations. For example,
suppose that we have two elements $\dfrac{t_1}{\prod_{(i,j) \in Y} p_{(i, j)}}$ and $\dfrac{t_2}{\prod_{(i, j) \in Y} q_{(i, j)}}$ of $\withdiv(T, Y)$ with all the expected requirements for $t$-s, $p$-s and $q$-s. Then, their sum is $\dfrac{t_1 \cdot \prod_{(i, j) \in Y} q_{(i, j)} + t_2 \cdot \prod_{(i, j) \in Y} p_{(i, j)}}{\prod_{(i, j) \in Y} (p_{(i, j)} \cdot q_{(i, j)})}$. Each product
$p_{(i, j)} \cdot q_{(i, j)}$ is indeed a polynomial from $\F[a_i, a_j]$. The numerator lies in $T$,
because $T$ is a subring of $\F[[a_1, a_2, \ldots, a_k]]$ that includes all polynomials; in particular, $T$ includes $\prod_{(i, j) \in Y} q_{(i, j)}$ and $\prod_{(i, j) \in Y} p_{(i, j)}$.
\end{proof}

Now, suppose that we want to proceed with the same type of the argument as in the previous
paper. So, we intersect our GF(2)-grammar with a DFA for $a_1^* a_2^* \ldots a_k^*$,
obtain a system of linear equations with indeterminates that correspond to the nonterminals of type $a_1 \to a_k$. As we recall, the system can be written down in the form $Ax = f$, where all entries
of $A$ lie in the field $R_{a_1, a_k} = \withdiv(S_2 (a_1, a_k), \{(1, k)\})$ (we are abusing
notation a bit here, but, hopefully, it should be clear what we are trying to say). 

On the other
hand, the entries of $f$ are complicated. Each entry of $f$ corresponds either to a ``completely
final'' rule that reduces some nonterminal to a terminal symbol, or to a ``simplifying'' rule of type
$A_{a_1 \to a_k} \to B_{a_1 \to a_m} C_{a_m \to a_k}$ for some $m$. ``Simplifying'' rules
correspond to much more complicated languages compared to ``completely final'' rules. Therefore,
$f$ can be written as $f_2 + f_3 + \ldots + f_{k - 1}$, with each $f_m$ corresponding to a sum
of products that split into a part from $1$ to $m$ and a part from $m$ to $n$. 

But, as we already know, that $x = A^{-1} f = \sum_{i=2}^{k-1} A^{-1} f_i$. It is
important to note that we have a sum here, meaning that we split into independent
summands. By continuing this sort of process recursively in the smaller nonterminal, 
we will represent $x$ as a sum of elements $\withdiv(S_k, Y)$, where all possible $Y$
can be described by the following recursive process:
\begin{enumerate}
\item On each step, we work with a subsegment $[\ell, r]$ of $[1, k]$. In any case,
add the pair $(\ell, r)$ to $Y$.
\item If $r - \ell = 1$, terminate the current branch of the recursion.
\item If $r - \ell \geqslant 2$, pick some $m$ from $\ell + 1$ to $r - 1$.
\item Recursively handle the segments $[\ell, m]$ and $[m, r]$.
\end{enumerate}
\begin{defn} Let us say that a subset $Y$ of $X_k$ is \emph{tree-like},
if it can be obtained by the above process.
\end{defn}
\begin{rem} In general, it is completely unsurprising that stratified sets have appeared here. They are absolutely the same as in the famous paper by Ginsburg and Ullian~\cite{unamb-class}, though here we care
only about the maximal ones. 
\end{rem}
\begin{ex} For $k = 2$, only $\{(1, 2)\}$ is tree-like.
For $k = 3$, only $\{(1, 2), (1, 3), (2, 3)\}$.
For $k = 4$, there are two: $\{(1, 2), (1, 3), (2, 3), (3, 4), (1, 4)\}$ and $\{(1, 2), (2, 3), (2, 4), (3, 4), (1, 4)\}$.
\end{ex}

We have just proven the following theorem:
\begin{thm}\label{gf(2)_tree} If $K \subset a_1^* a_2^* \ldots a_k^*$ is described by a GF(2)-grammar,
then $K$ can be represented as a sum, with at most one summand corresponding 
to each tree-like subset $Y$ of $X_k$. The summand that corresponds to $Y$
must lie in $\withdiv(S_k, Y)$. 
\end{thm}
\begin{rem}
In other words, $K$ lies in the $\F$-linear space generated
by $\withdiv(S_k, Y)$ with a tree-like $Y$.
\end{rem}

\section{A lower bound}

It turns out that the above upper bound is almost precise. To prove a corresponding 
lower bound, let us define an analogue of $\withdiv$ for invertible polynomials.
\begin{defn} Suppose that $T$ is some subring of $\F[[a_1, a_2, \ldots, a_k]]$ that contains
all polynomials. In other words, $\F[a_1, a_2, \ldots a_k] \subset T \subset \F[[a_1, a_2, \ldots, a_k]]$. Consider the set $X_k := \{(1, 2), (1, 3), \ldots, (1, k), (2, 3), \ldots, (2, k), \ldots, (k - 1, k)\}$ of all all pairs $(i, j)$, such that $1 \leqslant i < j \leqslant n$. Then, for every subset $Y$ of $X_k$, 
we can define a subset $\withinv(T, Y)$  of $\F[[a_1, a_2, \ldots, a_k]]$ in the following way. 
An element of $\F[[a_1, a_2, \ldots, a_k]]$ lies in $\withinv(T, Y)$ if and only if it can 
be represented as $\dfrac{t}{\prod_{(i, j) \in Y} p_{(i, j)}}$
for some $t \in T$ and all $p_{(i, j)}$-s being \emph{invertible} polynomials from $\F[a_i, a_j]$.
\end{defn}
\begin{lm} The set $\withinv(T, Y)$ is a ring.
\end{lm}
\begin{proof} The proof is the same as before. Instead of using that the product of polynomials is a polynomial, we will use that the product of \emph{invertible} polynomials is an invertible polynomial.
\end{proof}

\begin{thm} Suppose that $K \subset a_1^* a_2^* \ldots a_k^*$ can be represented as
a sum of elements, each coming from some ring $\withinv(S_k, Y)$ with a tree-like $Y$.
Then, there is a GF(2)-grammar for $K$. 
\end{thm}
\begin{proof} By closure under symmetric difference, it is enough to prove this theorem when
$K$ can be represented as $\dfrac{A_1 A_2 \ldots A_k}{\prod_{(i, j) \in Y} p_{(i, j)}}$ with a tree-like
$Y$, $A_i \in \A_i$, $p_{(i, j)} \in \F[a_i, a_j]$. 

Let us prove that via induction over $k$.
For the case of $k = 2$, the theorem has
already been proven in the previous paper~\cite[Theorem 8]{Makarov_DLT}. 
Now, suppose that $k \geqslant 3$. Then, because $Y$ is a tree-like set, there exists such $m$, that it splits the segment $[1, k]$ in the recursive definition of the tree-like sets. Then, 
$\dfrac{A_1 A_2 \ldots A_k}{\prod_{(i, j) \in Y} p_{(i, j)}}$ splits into $1/p_{(1,k)}$ and
a product of $\dfrac{A_1 A_2 \ldots A_m}{\prod_{(i, j) \in Y_\ell} p_{(i, j)}}$ and $\dfrac{A_{m + 1} \ldots A_k}{\prod_{(i, j) \in Y_r} p_{(i, j)}}$, where $Y_\ell \cup Y_r \cup \{(1,k)\} = Y$, all segments in $Y_\ell$ are subsegments of $[1, m]$ and all segments in $Y_r$ are subsegments of $[m, k]$. The former (respectively, the latter) of those two partial products is a subset of $a_1^* a_2^* \ldots a_m^*$ (respectively, $a_m^* a_{m+1}^* \ldots a_k^*$) that can be represented by a GF(2)-grammar. Hence, their product can be also represented by a GF(2)-grammar. Finally, 
we need to multiply said product by $1/p_{(1, k)}$. This can be done by creating a new starting
nonterminal $S'$ with rules $S' \to S$, $S' \to a_1^u S' a_k^ v$ for all monomials $a_1^u a_k^v$ in the non-invertible polynomial $p_{(1, k)} + 1$.
\end{proof}

\section{The linear case}

In the case of linear GF(2)-grammar, we can use the same line of reasoning, but
the final results will differ a bit because of the following two factors.
\begin{enumerate}
\item In the general case, we have a ``simplifying'' rule $A_{1 \to k} \to B_{1 \to m} C_{m \to k}$
that makes us ``recurse'' to both a nonterminal of type $1 \to m$ and to a nonterminal of type $m \to k$. In a sense, we have to make two ``recursive calls''. On the other hand, in the linear case, there is only one nonterminal on the right-hand side of each rule. Therefore, we will make
only one recursive call, replacing the ``tree structure'' observed in Theorem~\ref{gf(2)_tree} with 
a simpler ``path structure''. 
\item When we solve a linear system in the proof of Theorem~\ref{gf(2)_tree}, the coefficients
of the system are elements of the field $R_{a_1, a_k}$. However, in the linear case, the
coefficients of the system are from the ring $\F[a_1, a_k]$ of polynomials. Because
the latter ring is much simpler than the former field, we can easily prove a more precise
upper bound result. In the end, we obtain a \emph{completely precise} characterization 
for the subsets of $a_1^* a_2^* \ldots a_k^*$ described by linear GF(2)-grammars,
with the resulting upper and lower bounds matching exactly. 
\end{enumerate}

To make the idea of having 	``a path structure'' more precise, we will define the notion of
a \emph{path-like} set.
\begin{defn} Let us say that $Y \subset X_k$ is a \emph{path-like set}, if it can
be obtained by the following recursive process:
\begin{enumerate}
\item On each step, we work with a subsegment $[\ell, r]$ of $[1, k]$. In any case,
add the pair $(\ell, r)$ to $Y$.
\item If $r - \ell = 1$, terminate the current branch of the recursion.
\item If $r - \ell \geqslant 2$, recursively handle the segments $[\ell, r - 1]$ and $[\ell + 1, r]$.
\end{enumerate}
\end{defn}

\begin{thm}\label{linear_path} If $K \subset a_1^* a_2^* \ldots a_k^*$ is described by a linear GF(2)-grammar,
then $K$ can be represented as a sum, with at most one summand corresponding to each
path-like subset $Y$ of $X_k$. The summand that corresponds to $Y$ must
lie in $\withinv(\F[a_1, a_2, \ldots, a_k], Y)$. 
\end{thm}
\begin{rem} Notice that, unlike $\withdiv$ in the statement of Theorem~\ref{gf(2)_tree}, we have
$\withinv$ here, meaning that only invertible polynomials can appear in the denominators.
The are two other differences: a smaller ring $\F[a_1, a_2, \ldots, a_k]$ (instead of $S_k$)
``in the numerator'' and a different requirement ``for the denominators'' (the sets $Y$
have to be path-like and not tree-like).
\end{rem}
\begin{proof} Consider a linear GF(2)-grammar $G$ for the language $K$. Let us assume, 
without loss of generality, that $G$ is in the following normal form: each rule is either
$A \to \varepsilon$, $A \to c$, $A \to cB$ or $A \to Bc$. Then, let us formally intersect
$G$ with an automaton that accepts the language $a_1^* \{a_2, a_3, \ldots, a_{k-1}\}^* a_k^*$. Notice that this language is larger than $a_1^* a_2^* \ldots a_k^*$, but can be recognized
by a DFA with only three states $1$, $(2,k-1)$ and $k$.

Now, in the resulting GF(2)-grammar, we have rules $A_{1 \to k} \to a_1 B_{1 \to k}$, $A_{1 \to k} \to B_{1 \to k} a_k$, some ``completely final'' rules $A_{1 \to k} \to a_k$ and some ``simplifying'' rules that can only generate GF(2)-linear subsets of either $a_1^* a_2^* \ldots a_{k-1}^*$ or
$a_2^* \ldots a_k^*$. 
 
Similarly to the general case, we can
write down a linear system for all nonterminals of type $1 \to k$ of the resulting grammar. 
In the end, we will get a system $Ax = f$, where $f$ can be represented as a sum of two
summands $f_{1 \to k - 1}$ and $f_{2 \to k}$, each corresponding to a GF(2)-linear
grammar generating a subset of  $a_1^* a_2^* \ldots a_{k-1}^*$  or $a_2^* \ldots a_k^*$ respectively (``completely final'' rules can be ``absorbed'' into $f_{2 \to k}$). 

How do the entries of $A$ look like? Similarly to the general case, all the diagonal entries
are invertible polynomials and all other entries are not invertible. Hence, the determinant
of $A$ is an invertible polynomial. Therefore, by either Kramer's rule or the adjugate matrix
theorem, the coefficients of $A^{-1}$ lie in $\withinv(\F[a_1, a_k], \{(1, k)\})$. Or, in simpler terms, they are rational functions with invertible denominators. 

Finally, consider the equation $x = A^{-1} f = A^{-1} f_{1 \to {k - 1}} + A^{-1} f_{2 \to k}$.
Each of these two summands corresponds to one of the two recursive calls in the definition
of path-like subsets. The polynomials in the denominator will be invertible, because
the polynomials in the denominators of entries of $A^{-1}$ are invertible. Unlike the general
case, the numerators of entries of $A^{-1}$ are simple polynomials. In the end, these
three observations complete the proof.
\end{proof}

Now let us prove the converse result.
\begin{thm} Suppose that $K \subset a_1^* a_2^* \ldots a_k^*$ can be represented as
a sum of elements, each coming from some ring $\withinv(\F[a_1, a_2, \ldots, a_k], Y)$ with a path-like $Y$.
Then, there is a linear GF(2)-grammar for $K$.  
\end{thm}
\begin{proof} It is enough to prove this fact for each summand independently. Suppose that
we have a summand $K_Y$ that corresponds to a path-like set $Y$. We will
proceed by induction over $k$. 

Firstly, let us prove the statement for $k = 2$.
of $K_Y$ corresponds to some finite language. Suppose that it is generated by a grammar with
a starting nonterminal $S$. We need to ``multiply'' this language by $1/p_{(1, k)})$ for some
invertible polynomial $p_{(1, k)}$. This can be done by creating a new starting
nonterminal $S'$ with rules $S' \to S$, $S' \to a_1^u S' a_k^ v$ for all monomials $a_1^u a_k^v$ in the non-invertible polynomial $p_{(1, k)} + 1$.

The induction step is similar, we just have to take some measures to deal with the letters
$a_1$ and $a_k$ in the numerator. To deal with them, simply our summand into several
summands with the fixed powers of $a_1$ and $a_k$ in the numerator. For each of them,
construct the linear GF(2)-grammar inductively and add a new starting nonterminal $S'$
with a rule $S' \to a_1^\ell S a_k^m$ for the power $a_1^* a_k^*$. 
\end{proof}

\section{GF(2)-concatenation of GF(2)-linear languages}

The original paper about GF(2)-grammars~\cite{Makarov_DLT} claims that the family of languages 
described by \emph{linear} GF(2)-grammars is not closed under GF(2)-concatenation
and suggests the following proof plan. 

\begin{enumerate}
\item Find two unambiguous linear languages $L_1$
and $L_2$, such that their GF(2)-concatenation $L_1 \odot L_2$ seems to be too
``complicated'' to be described by a linear \emph{conjunctive} grammar.
\item Prove that there is no linear conjunctive grammar for $L_1 \odot L_2$ by
using the method of Terrier~\cite{Terrier}.
\item  Apply the fact that linear conjunctive grammars are strictly more powerful
than linear GF(2)-grammars and the fact that linear GF(2)-grammars are
strictly more powerful than unambiguous linear grammars~\cite{gf2}.
\end{enumerate}

If carried out, this plan will lead to the proof of the following conjecture:
\begin{conjecture} There are two unambiguous linear languages $L_1$ and $L_2$,
such that their GF(2)-concatenation $L_1 \odot L_2$ is not described by any linear 
conjunctive grammar.
\end{conjecture}
However, according to the best of authors' knowledge, the plan has not been carried
out anywhere in the literature yet. We suggest a much simpler proof of the original
statement by using Theorem~\ref{linear_path}. 

\begin{thm} Consider unambiguous linear languages $L_1 := \set{a^n b^n}{n \geqslant 0}$ and $L_2 := \set{b^m c^m}{m \geqslant 0}$. Their GF(2)-concatenation $L_1 \odot L_2 = \set{a^n b^{n + m} c^m}{n \geqslant 0 \text { and } m \geqslant 0}$ is not described by any linear GF(2)-grammar.
\end{thm}
\begin{rem} Notice that $L_1 \odot L_2 = \set{a^n b^{n + m} c^m}{n \geqslant 0 \text { and } m \geqslant 0}$ \emph{can} be described by a linear conjunctive grammar (or, equivalently, by a trellis automaton). In fact, it is one of the most well-known examples that show the power of trellis automata. 
\end{rem}
\begin{proof} Proof by contradiction. Suppose that $L_1 \odot L_2$ can be described by
a linear GF(2)-grammar. Then, by Theorem~\ref{linear_path}, the formal power series
$\sum_{n \geqslant 0, m \geqslant 0} a^n b^{n + m} c^m$
can be represented as 
\begin{equation}\label{f_lang}
\sum\limits_{n \geqslant 0, m \geqslant 0} a^n b^{n + m} c^m = \dfrac{\sum\limits_{i=1}^t p_i (a, b, c) G_i (a, b) + \sum\limits_{j = 1}^s r_j (a, b, c) H_j (b, c)}{1 + aq_a(a, c) + cq_c(a, c)}
\end{equation}
for some polynomials $p_i, r_j \in \F[a, b, c]$, some power series $G_i \in \F[[a, b]]$ and $H_j \in \F[[b, c]]$ and an invertible polynomial $1 + aq_a(a, c) + cq_c(a, c) \in \F[a, c]$. We are not even using the full strength of Theorem~\ref{linear_path} here, because we will not need any additional information about $G_i$ and $H_j$ in the rest of the argument. 

We can rewrite Equation~\eqref{f_lang} in the following way:
\begin{align*}
\sum_{i=1}^t p_i (a, b, c) G_i (a, b) + \sum_{j = 1}^s r_j (a, b, c) H_j (b, c) 
&= \\ (1 + ap_a(a, c) + cp_c(a, c)) \cdot 
\sum_{n \geqslant 0, m \geqslant 0} a^n b^{n + m} c^m
&= \\ (1 + ap_a(a, c) + cp_c(a, c)) \cdot \sum_{k=0}^{+\infty} b^k (a^k + a^{k-1}c + \ldots + ac^{k-1} + c^k) &= \sum_{k=0}^{+\infty} b^k d_k,
\end{align*}
where $d_k :=  (a^k + a^{k-1}c + \ldots + ac^{k-1} + c^k) \cdot  (1 + aq_a(a, c) + cq_c(a, c))$.

Now, consider the left-hand and the right-hand sides of this equation as a power series in variable $b$. By the right-hand side, the coefficient of said power series before $b^k$ is $d_k$. By the left-hand side, the powers of $a$ and $c$ in coefficient before $b^k$ are ``asymptotically separated''
in the following sense: there exists such $M$, that for any monomial $a^\ell b^k c^m$ in the left-hand side either $\ell$ or $m$ does not exceed $M$. 

Finally, let us think about the value $d_k$ for a bit. It has a lot of monomials of total degree $k$:
they correspond to the product of $a^k + a^{k-1}c + \ldots + ac^{k-1} + c^k$ and the 
summand $1$ from the sum $1 + aq_a(a, c) + cq_c(a, c)$. The rest of the monomials from $d_k$
are of total degree $k + 1$ or greater, because they correspond to the products of  $a^k + a^{k-1}c + \ldots + ac^{k-1} + c^k$ by monomials from $aq_a(a, c) + cq_c(a, c)$. Hence, they will
not cancel out with any of the summands $a^i c^{k - i}$. Therefore, $d_k$ contains monomials
$a^i c^{k - i}$ for all $i$ from $0$ to $k$. Finally, choose $i = k - i := M + 1$ or, in simpler terms,
$k := 2M + 2$, $i := M + 1$. Then, in the left-hand side, all monomials of the coefficient before $b^k$ do not exceed degree $M$ either in variable $a$ or in variable $c$. But in the right-hand
side we have a monomial $a^{M + 1} c^{M + 1}$ before $b^k$. Contradiction.
\end{proof}
\begin{rem} It is very important that the denominator in Equation~\eqref{f_lang} is an invertible polynomial. For example, consider a polynomial $a + c$ instead. Then, $d_k = (a + c)(a^k + a^{k-1}c + \ldots + ac^{k-1} + c^k) = a^{k + 1} + c^{k + 1}$, meaning that we have achived
a ``strict'' separation of $a$-s and $c$-s, let alone an ``asymptotic'' one.
\end{rem}

\section{Conclusion}

We have established an exact characterization of the subsets $a_1^* a_2^* \ldots a_k^*$ described
by \emph{linear} GF(2)-grammars and an \emph{almost} (as we have noted before, there is 
a small difference between the proven upper and lower bounds) exact characterization of 
the subsets $a_1^* a_2^* \ldots a_k^*$. Hence, we have almost fixed the ``imprecision'' of the
results of the previous paper~\cite[Section 5]{Makarov_DLT}. 

It would be nice to get rid of the word ``almost'' here and establish a complete characterization
for the case of general GF(2)-grammars as well. Hopefully, it is possible to reduce the 
remaining imprecision to checking validity of some purely algebraic statements, similarly
to one of the conjectures from the previous paper~\cite[Conjecture 2]{Makarov_DLT}.

\section{Acknowledgements}

This work is supported by Russian Science Foundation, project 23-11-00133.

\end{document}